\documentclass{sig-alternate}
\pdfoutput=1
\usepackage{graphics}
\usepackage{url}
\usepackage{subfigure}
\usepackage[lined,algoruled,linesnumbered,commentsnumbered]{algorithm2e}
\newtheorem{theorem}{Theorem}[section]

\allowdisplaybreaks

\numberofauthors{1}

\author{
Lei Tang\\
\affaddr{Mountain View,  CA 94041, USA}\\
\email{leitang@acm.org}
}
\title{Thresholding for Top-k Recommendation \\ with Temporal Dynamics}
\begin{document}
\maketitle

\begin{abstract}
  This work focuses on top-$k$ recommendation in domains where underlying data distribution
  shifts overtime.   We propose to learn a time-dependent bias for
  each item 
  over whatever existing recommendation engine. Such a bias learning
  process alleviates data
  sparsity in constructing the engine, and at the same time captures recent
  trend shift observed in data. 
 We present an alternating optimization
  framework to resolve the bias learning problem,
  and 
  develop methods to handle a variety of commonly used
  recommendation evaluation criteria,  as well as  large number of items
  and users in practice. 
  The proposed algorithm is examined, both offline and online,  using
  real world data sets collected from a retailer website. 
  Empirical results
  demonstrate that the bias learning can almost always 
  boost recommendation performance. We encourage other practitioners
  to adopt it as a standard component in recommender systems where
  temporal dynamics are a norm.

\end{abstract}

 \terms{Algorithms,
  Experimentation, Performance}


\section{Introduction}

Recommender systems have been extensively studied in different domains
including eCommerce~\cite{linden2003amazon}, movie/music
ratings~\cite{Kore-etal09-matrix}, news
personalization~\cite{Das-etal07-google},  content recommendation at
web portals~\cite{Agar-etal13-content}, etc.  And all sorts of methods
have been proposed for recommendation~\cite{Adom-Tuzh05-toward},
including content-based methods, neighborhood based
approaches~\cite{Desr-Kary11-comprehensive}, 
latent factor models like SVD or matrix
factorization~\cite{Kore-etal09-matrix}.  
We notice that many methods assume data in
recommender systems is
static or follows the same distribution. 
Hence, a significant body of
existing works evaluate their proposed models following a
cross-validation procedure by randomly hiding some entries in a
user-item matrix for testing. 

\begin{figure}[t]
  \centering \vskip -0.2in
  \includegraphics[width=0.4\textwidth]{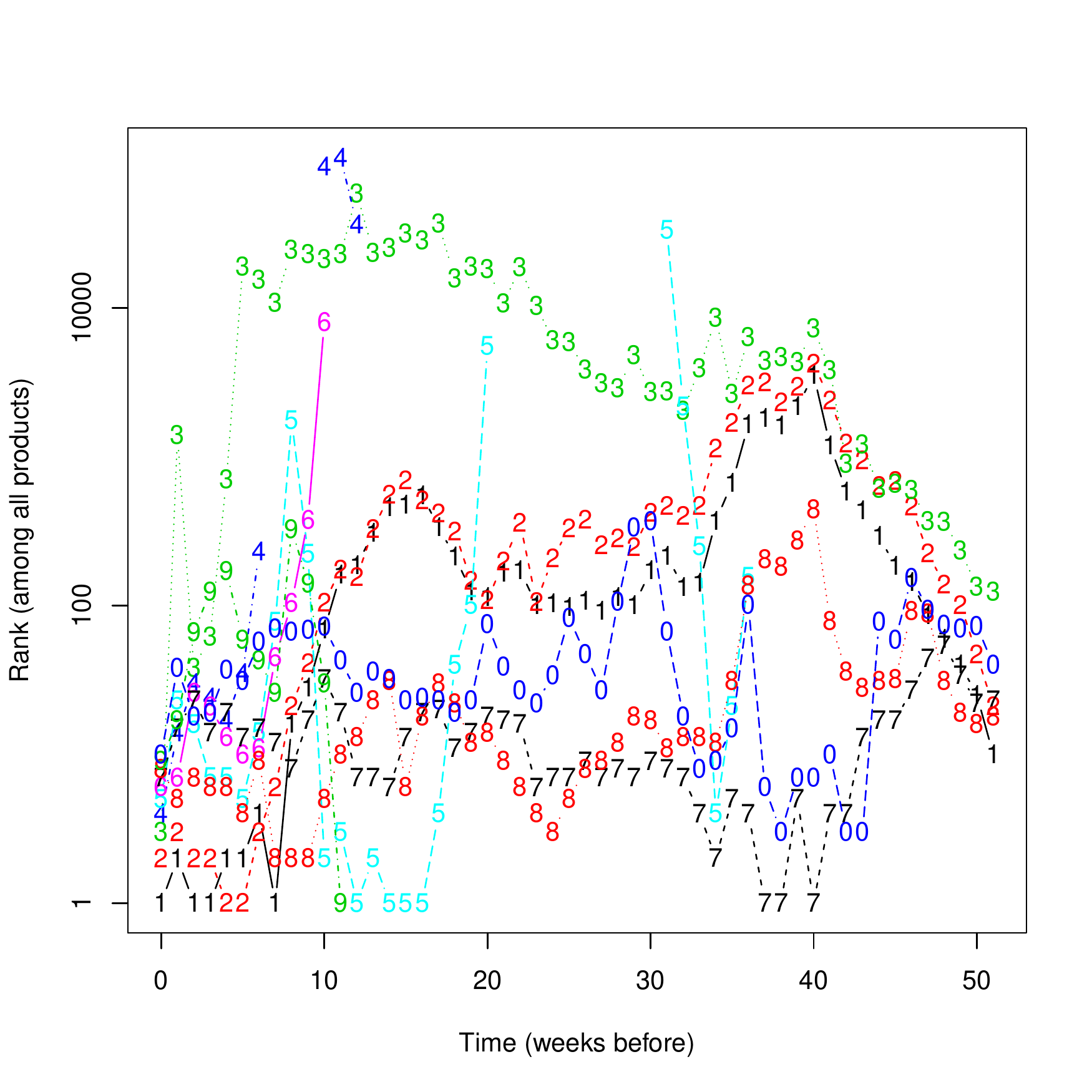}
  \vskip -0.2in
  \caption{Rank changes in past one year of top 10 popular products in
    early September, 2013. The relative ranking of products were
    determined by sales volume.  The number of each curve represent
    the product ranking at early september.  To make it legible, $0$
    indicates $10$. (Better viewed in color)}
  \vskip -0.2in
  \label{fig:rank}

\end{figure}

However, this is not the case for recommender
systems where user feedbacks are collected over time. In this paper,
we take one real-world eCommerce website (denoted as XYZ) as an example. 
The fundamental task is to recommend relevant items to customers,
hopefully to increase purchases and thus revenues and profits for the
company in the long run.  It is observed that the trending products
change week by week, or even day by day, due to user interest change,
demand shift ignited by some external events, or simply because a
product is out of inventory or just on shelf. 
Figure~\ref{fig:rank} demonstrates the fluctuation of
relative ranking (determined by sales volume) of the top 10 popular
products at XYZ 
identified at the first week of September, 2013.\footnote{All
privacy-sensitive items and services oriented products and warranties are removed for analysis.}  Notice that some
portion of the lines are missing for certain weeks in the figure,
indicating no transaction at the corresponding moments.  Among the 10
products, only 6 were sold one year ago, The other 4 products were not
even sold at the website yet.  The variance is huge as observed in the
figure. For instance, the 3rd popular product were ranked after
$10,000$ even just few weeks ago.  Majority of the select products,
did not enter top-$10$ mostly in the past.

The temporal dynamics described above pose thorny challenges for recommendation,
because it violates the fundamental assumption in most collaborative
filtering that \emph{training and test data share the same
  distribution}.  Simply ignoring this discrepancy will lead to bad
user experience . For instance, we may recommend a product that is
popular one year ago, but now has been replaced by new models.  This
is particular common for products on promotion during holiday season.
Another na\"{i}ve solution to deal with the discrepancy is to restrict
the model training to consider transactions only occurring the past few
days so that the test distribution would not change too much from that
of training.
However, that essentially discards plenty of past transaction
information for model training, leading to a more severe data
sparsity problem.

Data sparsity have been widely observed in many different domains.
At XYZ, for example, majority of customers purchase only few items across a whole year, and
majority of products have relatively small number of transactions.  It
is not a wise choice to discard data for model training. 
With fewer
data samples, the estimation of related parameters in the model will
be coupled with high variance. 
Moreover, a shorter time-window typically results in a
smaller coverage of items that can be leveraged for recommendation
since some items may not appear in the select window. In Figure~\ref{fig:markov_ri}, we plot the relative performance
improvement of a Markov model\footnote{Details and evaluation criteria are described later. }
when we expand the time window of training to past 1 month, 3 months,
6 months and 15 months respectively. The baseline is the model
trained using the past 1-week transactions only.  
Apparently, the larger the time window (and thus the more data)
we use for training, the better the performance is,  
implying that we should collect as
much data for training as possible.

\begin{figure}[t]
  \centering
  \includegraphics[width=.4\textwidth]{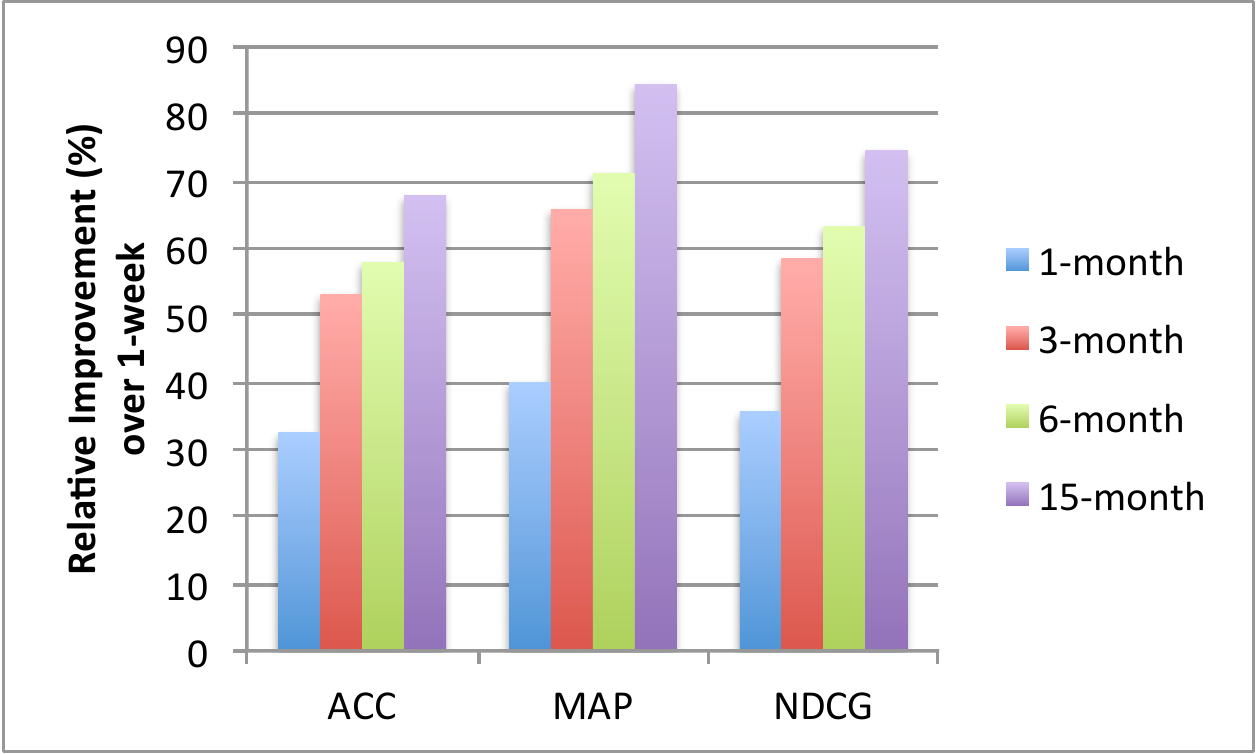}
  \vskip -0.1in
  \caption{Relative improvement of recommendation with different time
    windows for model training}
  \label{fig:markov_ri}
  \vskip -0.15in
\end{figure}


Hence we have the following dilemma: on one hand, we want to exploit
all possible transactions to overcome sparsity; on the other hand, we
wish to capture the trend change in recommendation by examining recent
transactions only.  How can we capture temporal dynamics without
discarding data for training?
In this work, we propose to learn
item-specific biases in order to capture the temporal trend change for
recommendation, while the recommendation model itself is still
constructed using as much data as possible.  We define a \emph{bias
  learning} problem and then propose one algorithm to optimize biases
for given evaluation criteria of top-$k$ recommendation.  Its
convergence properties and time complexity are carefully examined.
Experiments demonstrate the bias learning process will almost
always improve the performance of base recommendation model, by keeping pace with the temporal
dynamics of user-item interactions.

\section{Top-$k$ Selection and Evaluation}
Before proceeding to the bias learning problem, we need to review the
commonly used procedure in top-$k$ recommendation and corresponding
evaluation measures.  Without loss of generality, we will adopt
standard terms {\it user} and {\it item} to speak about the
recommendation problem throughout the paper .  We assume there are in
total $n$ users, $m$ items.  $u, v$ are used to denote the index for
users, $i, j$ for items, and $p,q$ the rank positions of items in
top-$k$.  More symbols and their meanings can be referred in
Table~\ref{tab:symbols}.

\subsection{Top-$k$ Selection}
In recommender systems, the number of items to recommend is typically
bounded by certain number.  For example, in email marketing, the
number of items can be recommended is limited by the email
template. 
Assuming that we need to select $k$ items for a given user $u$, a
common practice is to have an recommendation engine yielding a
prediction score for each item as follows:
\begin{eqnarray}
  \mathbf{f}_{u*}=  ( f_{u1}, f_{u2}, \cdots, f_{up})^T \label{eq:1}
\end{eqnarray}
where $f_{ui}$ denotes the prediction score of item $i$ for user $u$.
Note that some recommendation models yield scores only for a select
set of candidate items, with others default to $0$.
Then items are ranked according to their prediction scores and the top
ranking ones will be selected\footnote{We assume items are selected
  solely based on prediction scores, while researchers have been
  considering other factors like diversity~\cite{Yu-etal09-it}, which
  is beyond the scope of this work. }.
Equivalently, we find an \emph{ordered} set of $k$ items with maximal
scores:
\begin{eqnarray}
  \mathcal{R}(\mathbf{f}_{u*}) &=&  \{I^{(1)}, I^{(2)}, \cdots,
  I^{(k)}\}  \label{eq:topk-selection} \\
  & s.t. & f_u^{(1)} \ge f_u^{(2)} \ge \cdots \ge f_u^{(k)}, \nonumber \\ 
  &&  f_{u}^{(p)} \ge f_{u j},  \quad \forall p \in \{1, 2,
  \cdots, k\}, j \notin \mathcal{R}(\mathbf{f}_{u*}) \nonumber 
\end{eqnarray}
where $f_u^{(p)} = f_{u I^{(p)}}$ is the score of the top $p$-th item.

\subsection{Evaluation Criteria}
\label{sec:evaluation-criteria}
As for evaluation, the top ranking items
$\mathcal{R}(\mathbf{f}_{u*})$ and the relevance of items have to be
provided.  Different from typical ratings in collaborative filtering,
we focus on recommendations where there are only \emph{binary}
responses: relevant or irrelevant.  In our application, we deem one
item to be relevant for a user when the user purchases the item.  Let
$y_{ui} \in \{0, 1\}$ denote the relevance of item $i$ for a
particular user $u$, with $0$ being irrelevant, and $1$ relevant.
For presentation convenience, we use $y_u^{(p)}$ to denote the
relevance of the $p$-th top ranking item for $u$.

For top-$k$ recommendation, standard evaluation criteria include
\emph{accuracy} (ACC), \emph{mean average precision} (MAP), and
\emph{normalized discounted cumulative gain} (NDCG).  
Since all of them compute
performance with respect to top-$k$ recommendation of individual user and then average
across all users, we shall just describe them with respect to one
user.


Accuracy measures how many items in the top-$k$ recommendation are
indeed relevant.
\begin{equation}
  ACC@k (u) = \frac{1}{k}\sum_{ p =1}^k  y_u^{(p)}. \label{eq:ACC}
\end{equation}
\emph{Accuracy} does not care about position of items once  they
enter into top-$k$ recommendation.
By contrast, the other two measures \emph{average precision} ($AP@k$)
and \emph{normalized discounted cumulative gain} ($NDCG@k$) take into
account item position in the top-$k$ recommendation as well.

The precision up to rank $p$ (denoted as $Precision(p,u)$) is
essentially $ACC@p(u)$.
AP is the average of precision at positions of those relevant items:
\label{tab:utility-change}
\begin{eqnarray}
  \label{eq:AP}
  AP@k(u) & = &    \frac{\sum_{p=1}^k \left [ y_u^{(p)} \cdot Precision(p, u)
    \right ] } { \min
    \{ k, \text{\# relevant items for $u$}\}}.
\end{eqnarray}
The denominator in Eq.~(\ref{eq:AP}) is to account for cases when
users are associated with different numbers of relevant items.
Average precision is $1$ when all relevant items are ranked at the
top. When averaging AP@k for all users, we obtain \emph{mean average
  precision} (MAP@k).

Discounted cumulative gain (DCG) is initially proposed for rankings
with different degree of relevance:
\begin{equation}
  DCG @ k (u)  =\sum_{p=1}^k \left [   \frac{ 2^{y_u^{(p)}} - 1 }{
      log(1+p)  }\right ]   
  =  \sum_{p=1}^k \left [   \frac{y_u^{(p)} }{ log(1+j)  }\right ]  \label{eq:DCG}
\end{equation}
Eq~(\ref{eq:DCG}) follows because $y_u^{(p)} \in \{0, 1\}$ in our
case.
However, DCG generally varies with respect to $k$ and number of
relevant items for $u$ , making it difficult for comparison.  Hence,
\emph{normalized discounted cumulative gain} (NDCG) is proposed to
normalize the DCG into $[0, 1]$:
\begin{equation}
  NDCG @ k (u)  = \left [   \frac{y_u^{(p)} }{ log(1+j)  }\right ]
  /Z_{ku}, \label{eq:NDCG}
\end{equation}
where $Z_{ku}$ is the maximal DCG, i.e., the DCG when all relevant
items are ranked at the top.

Once we have the performance measure ($ACC@k$, $AP@k$, or $NDCG@k$)
for each individual user, the overall performance for a set 
users $\mathcal{U}$ can be computed as the mean:
\begin{equation}
  \label{eq:perf}
  perf@k(\mathcal{U})= \frac{1}{n}\sum_{u=1}^n  perf@k(u). 
\end{equation}
Since the ranking of items are determined based on score predictions,
and so are the corresponding performance, for presentation
convenience we can
rewrite 
Eq.~(\ref{eq:perf}) as
\begin{displaymath}
  perf(\mathbf{f}_{1*}, \mathbf{f}_{2*}, \cdots, \mathbf{f}_{n*}) = perf@k(\mathcal{U}), 
\end{displaymath}
where $\mathbf{f}_{u*}$ 
denotes the prediction scores for user $u$.

\section{The Bias Learning Problem}
As mentioned in introduction, recommendation in eCommerce, on one
hand, suffers from data sparsity, hence it is imperative to train the
recommendation model with transaction history of as long as possible.
On the other hand, the temporal dynamics of consumer purchase lead us
to weigh more for those items that are recently trending.  Such
fluctuation can be attributed to all kinds of factors.  For example,
retailer itself occasionally may post products with huge discount for
promotion.  External events like a nation-wide snow storm is very
likely to lead to booming purchases of heaters.
While capturing all varieties of external factors can be one way to
demystify the temporal dynamics, it is challenging to encompass them
all.  Alternatively, we take one data-based approach to learn a bias
for each item. In particular, we collect transactions in recent few
days or weeks, and attempt to determine item-specific biases such that
certain evaluation measure is maximized.
The problem can be formally stated as follows:
\begin{quote}
  {\bf Given}: \begin{itemize}
  \item  a set of users $\mathcal{U}$;
  \item  evaluation criterion $perf$;
  \item recent relevance information $Y$ of items for users in
    $\mathcal{U}$, with $y_{ui}$ denoting the relevance of item $i$
    for user $u$;
 \item prediction scores  $F$ for the set of users
    $\mathcal{U}$, with $f_{ui}$ being the score of item $i$ for $u$;
 \end{itemize}
 {\bf Find}: a vector of biases $\mathbf{b} = (b_1,
  b_2, \cdots, b_m)^T$ with $b_i$ indicating the bias of item
  $i$ so that
\begin{itemize}
  \item 
the top-k recommendation are selected as
    $\mathcal{R}(\mathbf{f}_{u*} + \mathbf{b})$ according to
    Eq.~(\ref{eq:topk-selection});
  \item the corresponding performance
    \begin{equation}
      perf(\mathbf{f}_{1*} +
      \mathbf{b}, \mathbf{f}_{2*} + \mathbf{b}, \cdots,
      \mathbf{f}_{n*} + \mathbf{b})       \label{eq:objective}
    \end{equation}
    is maximized.
  \end{itemize}
\end{quote}

Because the biases are learned with recent transactions only, it would
capture the recent trend change. The resultant biases can be different depending on the evaluation
criterion.  Note that all the three measures $ACC@k$, $AP@k$ and
$NDCG@k$ are not smooth with respect to prediction scores due to the
top-$k$ selection and dichotomy of relevance, making the problem
above intractable to solve
analytically.  
However, we shall show that the problem can be solved iteratively,
which is guaranteed to terminate in finite steps and converge to a
coordinate-wise (local) optimal.

\begin{table}
  \centering
  \caption{Nomenclature}
  \label{tab:symbols}
  \begin{tabular}{c|l}
    \hline
    Symbol(s) & Representation \\
    \hline 
    $n$ & total number of users \\
    $m$  & total number of items \\
    $k$  & number of items to recommend \\
    $u, v$ & index number of users \\
    $i, j$   & index number of items \\
    $p, q$ & index number of rank positions \\
    $f_{ui}$ & prediction score of item $i$ for $u$ \\
    $f_u^{(p)}$ & prediction score of the top $p$-th item for $u$ \\ 
    $\mathbf{f}_{u*}$ & prediction scores for  $u$ \\
    $\mathbf{f}_{*i}$ & prediction scores of item $i$ for all users \\
    $y_{ui}$ & relevance of item $i$ for $u$ \\
    $y_u^{(p)}$ & relevance of the top $p$-th item for $u$ \\
    $b_i$ & the bias of item $i$ \\
    \hline
  \end{tabular}
\end{table}

\section{Algorithm}
Since the bias learning aims to find a bias for each item, we can
rewrite the objective in Eq.~(\ref{eq:objective}) with respect to
items:
\begin{equation}
  \label{eq:objective_items}
  \max _{b_1, b_2, \cdots, b_m} perf(\mathbf{f}_{*1} + b_1, \mathbf{f}_{*2} + b_2,  \cdots,
  \mathbf{f}_{*m} + b_m). 
\end{equation}
The problem is difficult to resolve because of the ranking hidden in
calculating in $perf$. Yet, the problem is solvable if we optimize one
$b_i$ at a time.  We propose to adopt an \emph{alternating optimization}
approach.  That is, we fix the scores and biases for all other items
while optimizing the bias for one particular item.  We can cycle
through all items until the objective function is stabilized.  Next, we
first describe the case of finding the optimal bias for one single
item.  We'll use $ACC@k$ as an example to derive the algorithm and
then generalize it to handle other types of evaluation criteria like
$MAP@k$ and $NDCG@k$.

\subsection{Finding optimal bias for single item}
Keep in mind the accuracy contribution of one item solely depends on
whether the item enters into top-$k$. We start with the simplest case:
assuming for a given user $u$, item $i$ is not in top-$k$
recommendation, how will the objective in
Eq.~(\ref{eq:objective_items}) change accordingly if we increase the
score of $f_{ui}$ to push the item $i$ into top-$k$?  Once item $i$ is
pushed into top-$k$, naturally the top $k$-th item in the original
recommendation will be discarded.
This swap of items has four cases when considering the relevance of
each item, which is shown below:
\begin{equation}
  \label{eq:utility-change}
  \begin{array}[]{cc|c}
    \hline
    y_{ui} & y_u^{(k)} & \delta_{ui} =( y_{ui} - y_u^{(k)}
    )/ k \\
    \hline
    0 & 0 & 0 \\
    1 & 1 & 0 \\
    0 & 1 & -1/k \\
    1 & 0 & 1/k \\
    \hline
  \end{array}
\end{equation}
Apparently, when both items share the same relevance, that is, either
both relevant or irrelevant, the accuracy of user $u$ would not change
according to Eq.~(\ref{eq:ACC}).  The accuracy alters only if these
two items are associated with different relevance information.  In
particular, when the item discarded from top-$k$ is relevant while
item $i$ is not, the accuracy would decrease by $1/k$.  By
contrast, if item $i$ is relevant but the item discarded is not, the
accuracy increases by $1/k$.  We can summarize the accuracy change as
$( y_{ui} - y_u^{(k)} )/ k$ if we push item $i$ into the top-$k$
recommendation.

In order to push item $i$ into top-$k$ recommendation, the add-on bias
should be at least $f_u^{(k)} - f_{ui}$. Note that the bias value can be
negative.  Once the item enters top-$k$ recommendation,
the accuracy does not change even if we increase the bias more.  We
can record the necessary add-on bias of item $i$ to enter top-$k$
recommendation for each user, and compute its corresponding accuracy
change.  Then we can determine the overall accuracy change for all
users at different bias values and pick the optimal bias with maximal
utility increase.
For presentation convenience, we define \emph{utility as the accuracy
  change after item $i$ enters into top-$k$ recommendation}.  The
utility is $0$ if the item does not appear in the top-$k$
recommendation of any user, This serves as the reference point
(origin) in comparing utility values of different bias candidate
values.

\begin{algorithm}[t]
  \begin{small}
    \KwIn{item $i$, scores $\mathbf{f}_{*i}$, $\mathbf{f}_*^{(k)}$,
      and relevance $\mathbf{y}_{*i}$, $\mathbf{y}_*^{(k)}$\;}
    \KwOut{optimal bias $b_i$ and utility change $\Delta_i$\; }
    init candidate bias values $S = \phi$; $cur\_{util}=0$\;
    \For{each $u$}{ \If{$y_{ui} \ne y_{u}^{(k)}$} { compute utility
        change $\delta_{ui} = ( y_{ui} - y_u^{(k)} ) / k $\; compute
        score difference $s_{ui} = f_u^{(k)} - f_{ui}$\; append
        $(s_{ui}, \delta_{ui})$ to $S$\; \If{ $s_{ui} < 0$} {
          $cur\_{util} = cur\_{util} + \delta_{ui}$\; } } } push
    $(-inf, 0)$ to S\; find $b_i$ and $max\_util$ via subroutine in
    Algorithm~\ref{algo:subroutine}\; update $\Delta_i = max\_util -
    cur\_util$\; \Return{$b_i$, $\Delta_i$}
    \caption{Find optimal bias for item $i$}
    \label{algo:individual}
  \end{small}
\end{algorithm}
\begin{algorithm}[t]
  \begin{small}
    \KwIn{$cur\_util$, candidate pairs $S=\{ (s, \delta) \}$\;}
    \KwOut{optimal bias $b$ and $max\_util$\; } $S$ = sort $S$ by
    score difference $s$\; init $b=0$, $max\_util = cur\_util$\;
    init $util=0$\; \For{each $(s, \delta)$ in $S$}{ $util =
      util + \delta$\; \If{ $util > max\_util$}{ $b= s$\; $max\_util =
        util$\; } }
    \caption{Subroutine to find optimal bias given candidate <bias,
      utility change> pairs}
    \label{algo:subroutine}
  \end{small}
\end{algorithm}

Algorithm~\ref{algo:individual} presents the procedure to find the
optimal bias for one single item $i$ with respect to $ACC@k$.  Line
2-11 computes all possible bias values which may lead to a utility
change. Note that line 7-9 computes the current utility score when
bias is set to $0$. We'll prefer a $0$ bias if no utility improvement is
possible. 
Line 12 considers the extreme case that item $i$ does not
enter into top-$k$ recommendation for any user, i.e., we set the bias
to negative infinity (say, a huge negative constant), and the
corresponding utility is zero. This is included due to cases that item
$i$ entering into top-$k$ recommendation leads to a negative utility.
Then it is better to exclude item $i$ from top-$k$ recommendation.
Once we have all the potential bias values that may result in a
different utility score, the subroutine in
Algorithm~\ref{algo:subroutine} sorts the bias values in ascending
order, and figure out the optimal bias which leads to the maximal
utility.  Then in line 14-15 of
Algorithm~\ref{algo:individual}, we report the optimal bias and corresponding utility
change comparing with current utility.

The time complexity for such a procedure is $O(n\log n)$ where $n$ is
the number of users. Line 2-11 in Algorithm~\ref{algo:individual} cost linear time. Algorithm 2 costs
$O(n\log n)$ because of the sorting in line 1.  Therefore, the total
computational time to find the optimal bias for one single item is
$O(n \log n)$.

\subsection{Finding optimal biases for all items}
In the previous subsection, we have present an algorithm to learn the
optimal bias of one single item when fixing the scores for all other
items.  In order to find optimal biases for all
items, 
we propose to cycle through all items to update biases, tantamount to
the well-known coordinate descent method in optimization.
The detailed algorithm in shown in Algorithm~\ref{algo:coordinate}.

\begin{algorithm}[t]
  \begin{small}
    \KwIn{score predictions $\{\mathbf{f}_{1*}, \mathbf{f}_{2*},
      \cdots, \mathbf{f}_{n*}\}$, relevance
      information$\{\mathbf{y}_{1*}, \mathbf{y}_{2*}, \cdots,
      \mathbf{y}_{n*}\}$\;} \KwOut{ optimal biases
      $\mathbf{b}^*=\{b_1^*, b_2^*, \cdots, b_m^*\}$} construct
    candidate item set $C$ to learn optimal bias\; init $b_i^*=0$ for
    all $i \in C$\; init $\Delta_{utility} = 1$\;
    \While{$\Delta_{utility}>0$ } { reset $\Delta_{utility}=0$\;
      \For{each item $i$ in $C$} { compute potential utility change
        $\Delta_i$ and corresponding bias $b_i$ following
        algorithm~\ref{algo:individual}\; \If { $\Delta_i> 0$ } {
          update $\Delta_{utility} = \Delta_{utility} + \Delta_i$ \;
          update $b_i^* = b_i^* + b_i$\; update prediction scores
          $\mathbf{f}_{*i}$ for item $i$\; update top $k$-th
          recommendation for each user if necessary\; } } }
    \caption{Learning Biases for All Items}
  \end{small}
  \label{algo:coordinate}
\end{algorithm}

The algorithm consists of two loops: the inner loop cycles through all
items to find the optimal bias individually. The outer loop stops if
no utility increase can be found.  The related update in lines 8-12
after changing the bias of one item is straightforward except line
12. 
In order to save computational cost, we can record the top-k
recommendation given current scores, and the minimum score in the
top-k recommendation for each user. Based on the minimum score, we can
immediately decide whether the item $i$ is in the top-k
recommendation before and after bias is updated.  If the item $i$ is
not among top-$k$ recommendation before and after bias update, then
nothing needs to be done.  We can take care of other cases in a
similar vein so that we do not have to recompute top-$k$
recommendation for each user, and thus speed up the computation
significantly.

Based on the algorithm above, we can derive theoretical properties
below:
\begin{theorem}\label{thm:converge}
  Algorithm~\ref{algo:coordinate} is guaranteed to terminate in finte
  steps.
\end{theorem}
\begin{proof}
  The $\Delta_{utility}$ is upper-bounded by the maximum accuracy of
  top-k recommendation among all the possible rankings.  Each
  cycle (lines 4 - 15) in the algorithm will increase utility by
  at least $1/k$ if not zero.  Therefore, the algorithm must terminate
  in finite iterations.
\end{proof}

The number of cycles tends to be very small in reality. In most
cases, we just need to cycle through all items a couple of times.
Nevertheless, we show that based on the algorithm, we are able to
filter out certain items so that each iteration scans smaller number of
items, which is stated below:
\begin{theorem}\label{thm:label}
  If one item $i$ satisfies $y_{ui} = 0$ for all $u$ , then the item
  $i$ can be removed from consideration of recommendation.
\end{theorem}
\begin{proof}
  Let's revisit the utility change table as shown in
  Eq.~(\ref{eq:utility-change}).  Note that when $y_u^{(k)}$ is $0$,
  swapping in a different item in the top-k recommendation will always
  increase or keep current utility.
  Because item $i$ satisfies $y_{ui}=0$ for all $u$, it follows that
  setting a bias of negative infinity for item $i$ would lead to no
  utility loss overall (if not positive utility change), no matter
  whether item $i$ is in the top-$k$ recommendation.  It is
  essentially equivalent to removing item $i$ from consideration of
  recommendation.
\end{proof}

The theorem above suggests that we may remove items without positive
relevance from the input score predictions in
Algorithm~\ref{algo:coordinate}.  In the context of eCommerce, those
products which does not appear in recent transactions can be removed
from recommendation directly and thus reduce computational cost
substantially.

The proposed algorithm is sequential in nature, and thus difficult to
parallel.  The time complexity for the proposed algorithm is at least
$O(mn\log n)$ where $m$ is the size of candidate item set $C$, and $n$
is the number of users.  Even though with few cycles, the time and
space complexity will be scary when both numbers of users and items
are huge.  Next we discuss a couple of heuristics we might consider in
practice.  All these methods have been adopted in our experiments and
have been shown to save tremendous time and space.

\subsection{Scaling up for Practical Implementation}
\label{sec:practical-implementations}

First of all, it is not necessary to store the prediction scores for
all items.  Considering recommendation for even just 100K users with
100K items, which is medium size in the era of big data. Assuming each
score takes 4 bytes as a float number, it requires $100K \times 100K
\times 4 = 40G$ memory space.
Therefore, we suggest exploiting a sparse representation of prediction
scores by keeping only select number of top recommendations while
defaulting others to zero. The number of recommendations to keep
is typically a multiplier of $k$.  For instance, when we need to
optimize for performance@top-10, we can keep the top-$50$
recommendations from a model. For the remaining items, their predict
scores are set to $0$.  This is valid because we often observe a fast decay
of the scores no matter what the base recommendation model is.


For another, we can shrink the candidate item set for bias tuning to
reduce time complexity.
Two types of items should be considered with higher priority for bias
tuning.  One set are those items which are recommended frequently based
on raw prediction scores, which tend to be past popular items.  The
other set are those appearing frequently in recent transactions, which
are the recently popular ones.  The former is likely to lead to a
negative bias, while the latter a positive bias.  For other
items outside the candidate set, we set their bias value to
zero without changing their prediction scores.

Moreover, we notice that the number of items with updated bias scores
is dramatically decreased with iterations of cycling through items.
The algorithm reaches a cooridate-wise local optimal and stops after
few cycles.  Yet, the scanning of all items are expensive, and
thus early stopping criteria can be used. For example, we may set a
percentage threshold such that if fewer than the percentage of items
need to update bias, then we terminate the cycling.  Or we can
simply set the upper-bound of cycles.  In our experiments, we just
set it to $2$ to reduce computational time, yet found no performance
loss.

\subsection{Extensions to optimize MAP or NDCG}
\label{sec:map-ndcg}
We have described how we optimize biases with respect to $ACC@k$. Now
we extend the algorithm~\ref{algo:individual} to optimize $MAP@k$ or
$NDCG@k$.
Different from accuracy, for which the position of one item in the
top-$k$ recommendation does not matter, MAP and NDCG is
\emph{rank}-related metric.  The position of one item in the top-$k$
recommendation plays an import role. A relevant item ranked as top-1
will results in a different score than that when the item is ranked as
the top $k$-th.  In order to compute the potential utility change with
different bias scores, we have to consider all the possible positions
rather than just the top $k$-th item.  Nevertheless, the basic idea
remains the same.  
We
compute out the score difference and corresponding utility change
with respect to each position in top-$k$.

Take average precision in Eq.~(\ref{eq:AP}) as an example.  Given a
user, assume the average precision of top-$k$ recommendation without
item $i$ is $AP_0$, serving as our reference point. We can gradually
increase the score of item $i$ to be ranked at position $k$, $k-1$,
$\cdots$, $1$, and compute out the corresponding performance.  For
each position, its utility change should be computed as follows:
\begin{center}
  \begin{tabular}{c|c|c}
    \hline
    Position & Performance & utility change \\
    \hline
    $k$ & $AP_k$ & $\delta_k = AP_k - AP_{0}$ \\
    $k-1$ & $AP_{k-1}$ & $\delta_{k-1} = AP_{k-1} - AP_{k}$\\
    $\vdots$ & $\vdots$ & $\vdots$ \\
    $1$ & $AP_{1}$ & $\delta_{1} = AP_{1} - AP_{2}$ \\
    \hline
  \end{tabular}
\end{center}
Suppose item $i$ is currently at position $p+1$, and we increase its
bias to position $p$.  Note that only items at position $p$ and $p+1$
are swapped, while the other items remain unchanged. Therefore, we only
need to examine the utility change because of the swap of the two
items.  If item $i$ and the item at position $p$ have the same
relevance, i.e.  $y_{ui} = y_{u}^{(p)}$, the AP would not change, and
hence leading to no utility change.  If $y_{ui} = 1$ and $y_{u}^{(p)}
= 0$, then it leads to a utility increase when we push item $i$ into
position $p$; If $y_{ui} =0$ and $y_{u}^{p} = 1$, then a reduction of
utility.  Let
$\tilde{k} = \min \{k, \# \text{relevant items for}~u\}$, the utility
change can be derived below:
\begin{eqnarray}
  \delta_{ui}^{(p)} = ~~~~~y_{ui}\cdot \left ( \frac{1+ \sum_{q=1}^{p-1}y_u^{(q)}}{p} - 
    \frac{1+ \sum_{q=1}^{p-1}y_u^{(q)}}{p+1}\right ) / \tilde{k} \label{eq:map-term1}\\
  +~y_{u}^{(p)}\cdot \left ( \frac{1+ \sum_{q=1}^{p-1}y_u^{(q)}}{p+1} - 
    \frac{1+ \sum_{q=1}^{p-1}y_u^{(q)}}{p}\right ) / \tilde{k} \label{eq:map-term2}\\
  =~~~\left ( y_{ui} - y_{u}^{(p)} \right )\left ( 1+ \sum_{q=1}^{p-1}y_u^{(q)}
  \right )\left (\frac{1}{p} - \frac{1}{p+1} \right ) /
  \tilde{k} \label{eq:map-change} \\
  with~~~~\frac{1}{p+1} = 0 ~~~if~p+1 > k. ~~~~~~~~~~~~~~~~~~~~~~~~~ \nonumber
\end{eqnarray}
In the equation, the first term in Eq.~(\ref{eq:map-term1}) is the
utility change that item $i$ moves from position $p+1$ to position
$p$, and the second term in Eq.~(\ref{eq:map-term2}) is the utility
change when the original item at position $p$ downgrades to position
$p+1$.  There is one special case when $p+1 > k$, i.e., item $i$ has
not entered into top-$k$ yet at the beginning.  In that case, we
replace $1/(p+1)$ by $0$.   Similarly for NDCG, it follows that
\begin{eqnarray}
  \label{eq:ndcg-change}
  \delta_{ui}^{(p)}=\left ( y_{ui} - y_{u}^{(p)} \right )\left (\frac{1}{\log
      (1+p)} - \frac{1}{\log (1+ (p+1))} \right ) /
  Z_{uk}  \\
  with~~~\frac{1}{\log (1+ (p+1))} = 0 ~~~if~p+1 >k. ~~~~~~~~~~~~~~~~ \nonumber 
\end{eqnarray}

Note that when $y_{ui}$ and $y_u^{(p)}$ are the same, both utility
changes following Eqs~(\ref{eq:map-change}) and (\ref{eq:ndcg-change})
would be zero.  Hence, rather than checking every possible position in
top-$k$, we just need to check those positions with different
relevance to item $i$.  The algorithm to find optimal bias for one
single item is summarized in Algorithm~\ref{algo:NDCG}.  It is almost
the same as Algorithm~\ref{algo:individual}, except line 3-5.  Line
3-4 are supposed to check each potential position in the top-k, and
line 5 is to update the utility change based on corresponding
performance metric. Apparently, such a change leads to an increase of
time complexity, reaching $O(kn\log kn)$ to find optimal bias for
single item. This is still fine if $k$ is
reasonably small, which is mostly true in practice. Plugging 
Algorithm~\ref{algo:NDCG} into Algorithm~\ref{algo:coordinate},  we obtain the
bias values for all items with respect to select evaluation criterion.

\begin{algorithm}[t]
  \begin{small}
    \KwIn{item $i$, scores $\mathbf{f}_{*i}$, $\mathbf{f}_*^{(k)}$,
      and relevance $\mathbf{y}_{*i}$, $\mathbf{y}_*^{(k)}$\;
      rank-related performance metric} \KwOut{optimal bias $b_i$ and
      utility change $\Delta_i$\; } initialize candidate set $S =
    \phi$, $cur\_{util}=0$\; \For{each $u$}{ \For{each position $p$}{
        \If{$y_{ui} \ne y_{u}^{(k)}$} { compute utility change
          wrt. performance metric following Eq.~(\ref{eq:map-change})
          or (\ref{eq:ndcg-change})\; compute score difference $s_{ui}
          = f_u^{(k)} - f_{ui}$\; append $(s_{ui}, \delta_{ui})$ to
          $S$\; \If{ $s_{ui} < 0$} { $cur\_{util} = cur\_{util} +
            \delta_{ui}$\; } } } } push $(-inf, 0)$ to S\; find $b_i$
    and $max\_util$ via subroutine in
    Algorithm~\ref{algo:subroutine}\; update $\Delta_i = max\_util -
    cur\_util$\; \Return{$b_i$, $\Delta_i$}
    \caption{Find optimal bias wrt. MAP/NDCG}
    \label{algo:NDCG}
  \end{small}
\end{algorithm}

\section{Experiment Setup}
In this section, we mainly describe the basic setup for our
experiments, including preparation of benchmark data sets, base
recommendation model and other methods considering temporal dynamics
for comparison.

\subsection{Benchmark Data Sets}
We collect customer transactions of XYZ 
 and construct
benchmark data sets via a split based on date. User activities before
the date are used for training, and the transactions in the subsequent
week are used to evaluate recommendation performance.  Corresponding
performance measures include ACC@k, MAP@k, and NDCG@k as described in
Section~\ref{sec:evaluation-criteria}.  In our experiments, $k$ is set
to $10$. 
 For easy interpretation, we report all numbers in
terms of lift (relative improvement) with respect to one baseline:
\begin{displaymath}
  lift = \left ( \frac{perf}  { perf_{baseline} }- 1 \right ) * 100\%.
\end{displaymath}

We prepare two benchmark data sets: one is during regular season
(Nov. 1, 2013); and the other is during holiday season (Dec. 11,
2013).
User shopping behavior in holiday season tends to be quite different
from regular season, both in terms of quantity and trending products.
We aim to verify the efficacy of our proposed method under both
settings.

\subsection{Base Recommendation Model}
One commonly used approach for recommendation in eCommerce is to model
user actions as a Markov chain
\cite{Shan-etal05-an,Rend-etal10-factorizing}. It is tantamount to
computing item similarity~\cite{linden2003amazon}, but keeping the metric directional.
That is, one's current action depends only on his most recent
action. 
The transition probability from one action to
another can be estimated below: 
\begin{equation}
  P(\text{buy}~i|\text{bought}~j) =  \frac{ \text{\# users who bought
      $j$ then $i$} } { \text{\# users who bought $j$} }. \label{eq:1st-order}
\end{equation}
We also take into consideration of those highly associated purchase
patterns by computing co-purchase probability of multiple items beyond
$2$.  For instance, the transition of two actions leading to one
purchase can be computed as:
\begin{equation}
  P(\text{buy}~i|\text{bought}~j_1, j_2) =  \frac{ \text{\# users
      who bought $j_1, j_2$ then $i$} } { \text{\# users who bought
      $j_1$, $j_2$} }. \label{eq:2nd-order}
\end{equation}
However, considering those higher-order purchase patterns leading to
explosion of state spaces. Therefore, we consider only the state
spaces containing up to $2$ actions.
As for prediction, we pick the products with highest probability given
user's most recent few transactions.

This Markov model is exploited because it has been validated to work
quite well in eCommerce. Later in experiments, we shall show that our
proposed bias learning can be applied to other base recommendation
models as well.




\subsection{Methods Considering Temporal Change}
For our proposed bias learning method (denoted as $M_{bias}$), we use
most recent 3-day transactions to fine tune the bias of items.  
All the biases, unless specified, are learned via optimizing ACC@k.
As for comparison, we also include one baseline method without
considering temporal change.
\begin{description}
\item $M_{long}$: This method utilizes as long history as possible for
  training. As already shown in Figure~\ref{fig:markov_ri}, the longer
  time window we use, the better the recommendation model performs. In
  our experiments, we use up to 15 months
of user activity history for
  training.
\end{description}
Besides the baseline, there are several other approaches to take into
account temporal dynamics.
\begin{description}
\item $M_{truncate}$: According to Theorem~\ref{thm:label}, if one item
  does not appear in recent transactions, we can remove it from
  recommendation.  This method trains the recommendation model using
  15-month data, but for prediction it concentrates only on those
  items that are recently attracting user attentions.  It truncates
  the item set for recommendation but does not tune biases of items.
\item $M_{distrdiff}$: This method directly computes a bias term for
  each item, rather than optimizing biases with respect to certain
  metric.  In particular, we compute two distributions of items, one
  from the 15-month transactions(denoted as $\mathbf{d}^{(long)}$), and the
  other from recent transactions (denoted as $\mathbf{d}^{(short)}$). So the
  item-specific bias is $\mathbf{b} = \mathbf{d}^{(short)} - \mathbf{d}^{(long)}$.
  Biases are added to the normalized probabilistic output of recommendation engine
  to promote recent trending products while suppressing past popular
  ones.
\item $M_{decay}$: Another option is to train a model that already
  incorporates temporal dynamics, by assigning lower weight to those
  remote events.  An exponential decay function is used:
   $ w = \exp { ( - \Delta_t / \beta ) }$
  where $\Delta_t$ is the time gap of the purchase in
  Eqs~(\ref{eq:1st-order}) and (\ref{eq:2nd-order}) to current date
  of recommendation, and $\beta$ is a decay factor.  $\beta$ is
  set to 60 in our experiments.
\end{description}

Note that there has been some work to consider temporal dynamics for recommendation. For example, Koren et al.~\cite{Kore09-collaborative,Koen-etal11-Yahoo}
proposed to have a time-dependent bias in matrix factorization for
movie/music ratings. The model minizes the root mean squared error
and adopts stochastic gradient descent to
find biases and latent factors. However, in our application, the
responses are binary (either purchase or not purchase), and only
positive responses are collected.  A trivial solution would set bias
to $1$ for all, which is meaningless.  We may randomly
sample negative entries for our one-class collaborative
filtering~\cite{Pan-etal08-one} problem, but that would essentially
connect bias to the sampling rate, which is not acceptable either. 

\section{Experiments}
In this section, we conduct a series of experiments over the
constructed benchmark datasets to study the performance, sensitivity
to base recommendation models and performance metrics.  At the end, we
report results by applying our method to online A/B test.

\subsection{Performance Comparison}
The performance of various methods are shown in
Tables~\ref{tab:mkv-regular-lift} and \ref{tab:mkv-holiday-lift}.  For
easy comparison, we deem $M_{long}$ as the baseline and show the lift
of other methods.  The numbers in bold face indicates the one with
best performance. For both data sets at regular season and holiday
season, our proposed method $M_{bias}$ is the winner.  The lift is
more observable at holiday season because of the strong shopping
pattern change thanks to Black Friday, Cyber Monday and other
promotion campaigns.  The numbers might look small, but keep in mind
it took nearly 3 years and thousands of teams worldwide to improve
10\% over a trivial baseline in Netflix prize competition\footnote{\url{http://en.wikipedia.org/wiki/Netflix_prize}}. 
 It is noticed that $M_{truncate}$ and
$M_{distrdiff}$ both yield some improvement, suggesting that it is
always helpful to incorporate recent trend. However, neither of them
is comparable to learning a bias for each item as we proposed.
As shown in both tables, adding a weighted decay based on recency for
training does not help.  In short, learning biases to
capture the temporal dynamics can model individual interests and
preferences more accurately, and thus improve performance of the base
recommendation model.  This is especially helpful when the temporal
fluctuation is huge, as shown during the holiday season.

\begin{table}[t]
  \vskip -0.1in
  \centering
  \caption{Lift (\%) at Regular-Season Data}
  \label{tab:mkv-regular-lift}
  \begin{tabular}{l|rrr}
    \hline
    Method   & ACC@10 & MAP@10 & NDCG@10 \\
    \hline
    $M_{long}$ & 0 & 0 & 0\\
    $M_{bias}$ & {\bf 1.228} & {\bf 0.842} & {\bf 0.972} \\
    $M_{truncate}$ & 0.460 & 0.142 & 0.239 \\
    $M_{distrdiff}$ & 0.350 & 0.384 & 0.389 \\
    $M_{weighted}$ & -10.703 & -11.083 & -10.495 \\ 
    \hline
  \end{tabular}
  \caption{Lift (\%) at Holiday-Season Data}
  \label{tab:mkv-holiday-lift}
  \begin{tabular}{l|rrr}
    \hline
    Method   & ACC@10 & MAP@10 & NDCG@10 \\
    \hline
    $M_{long}$ & 0 & 0 & 0\\
    $M_{bias}$ & {\bf 5.857} & {\bf 5.391} & {\bf 5.482} \\
    $M_{truncate}$ & 1.058 & 0.806 & 0.860 \\
    $M_{distrdiff}$ & 0.921 & 0.593 & 0.705 \\
    $M_{weighted}$ & -4.737 & -4.302 & -4.174 \\ 
    \hline
  \end{tabular}
  \vskip -0.1in
\end{table}

\begin{table}
  \centering
  \vskip -0.1in
  \caption{Coverage of Top Popular Items}
  \label{tab:top-popular}
  \begin{tabular}{|c|cc|}
    \hline
    top popular items    & without bias & with bias  \\
    \hline
    10 & 3 & 7 \\
    20 & 8 & 15 \\
    50 & 24 & 28 \\
    100 & 48 & 57 \\
    200 & 97 & 116 \\
    500 & 260 & 300 \\
    1000 & 531 & 576 \\
    \hline
  \end{tabular}
\vskip -0.2in
\end{table}

\begin{figure*}[t]
  \begin{minipage}[b]{0.34\textwidth}
    \includegraphics[width=\textwidth]{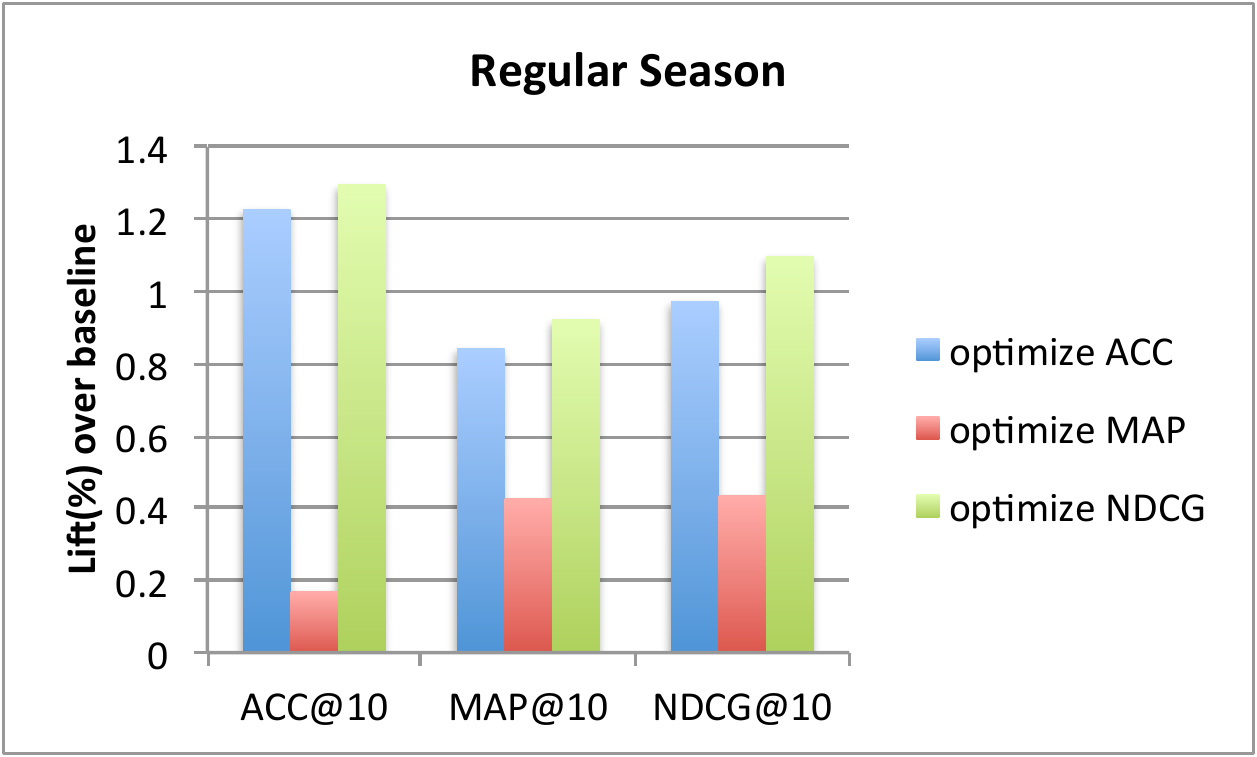}
    \caption{Lift at Regular Season}
    \label{fig:regular-map}
    \end{minipage}
    \begin{minipage}[b]{0.34\textwidth}
    \includegraphics[width=\textwidth]{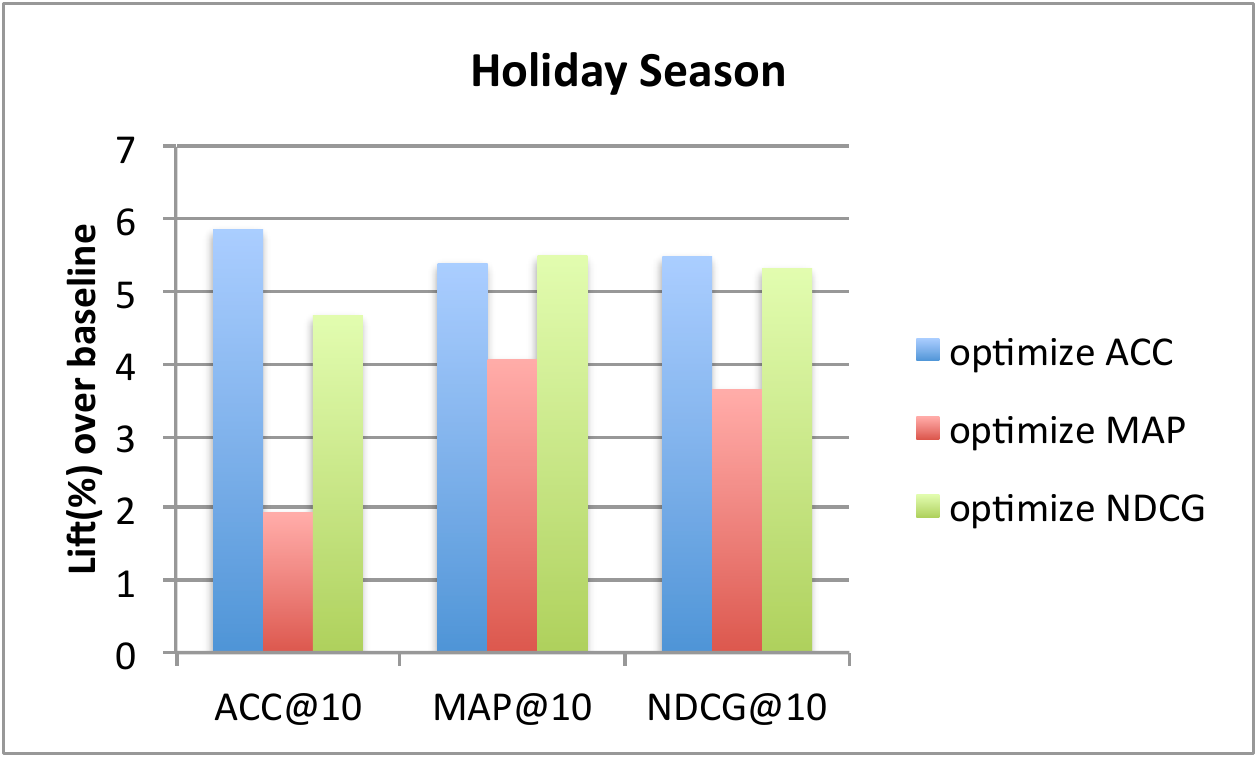}
    \caption{Lift at Holiday Season}
    \label{fig:holiday-map}
    \end{minipage}
    \begin{minipage}[b]{0.32\textwidth}
    \includegraphics[width=\textwidth]{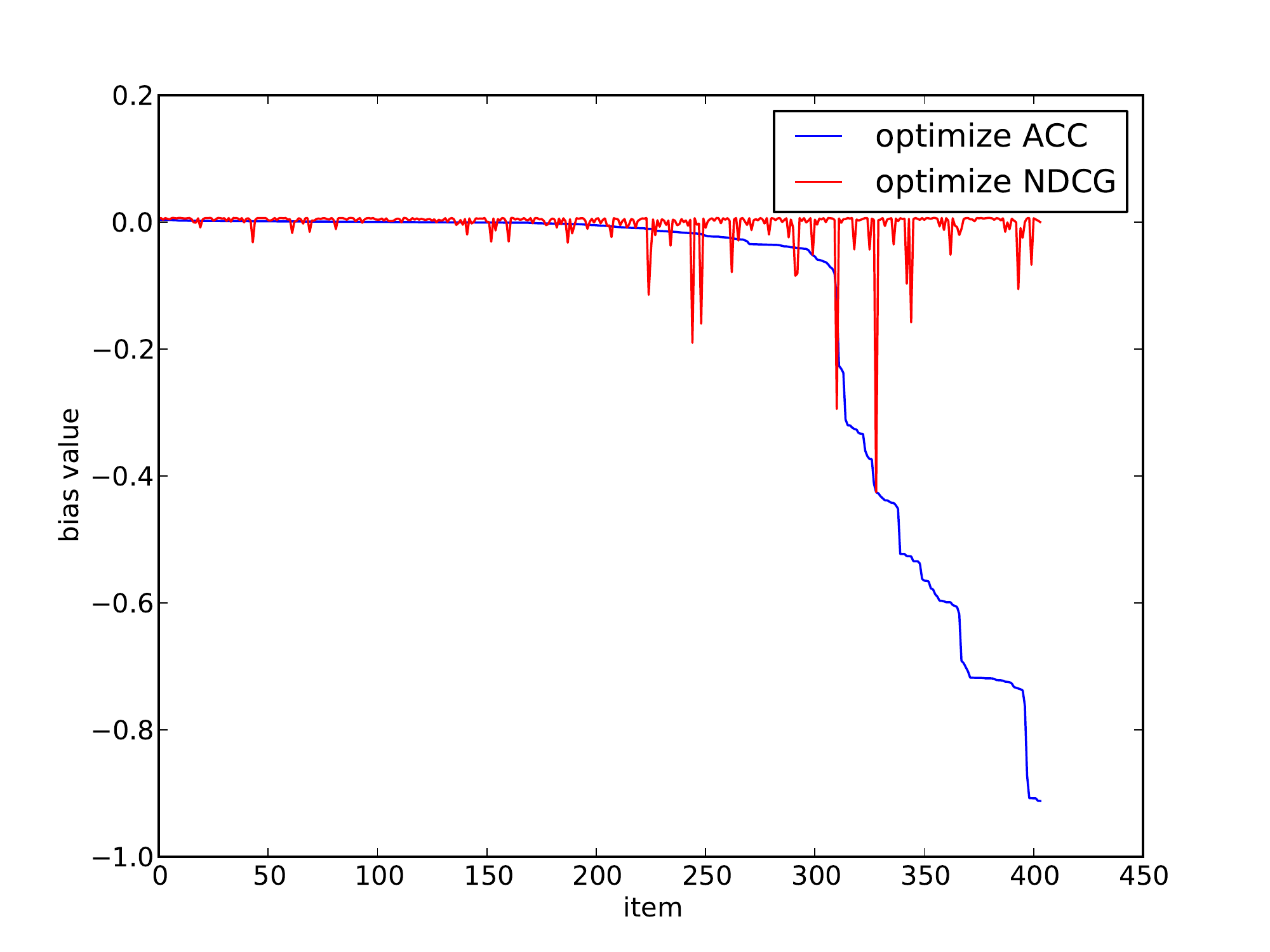}
    \vskip -0.1in
    \caption{Bias Values}
    \label{fig:bias-values}
    \end{minipage}
    \vskip -0.15in
  \end{figure*}

At macro level, we observe that bias learning tends to shift the item
distribution in our top-k recommendation towards the ground truth.  We sort all
items based on their frequency in reverse order in test data and
predictions respectively, and then check the overlap of top popular
items between the two. The result over the holiday season data is
shown in Table~\ref{tab:top-popular}.  A similar trend is observed
during regular season.  Apparently, adding the bias helps the
prediction to capture those recently trending products.  For example,
only 3 out of the top-10 popular products in testing period is covered
by the raw recommendation.  However, adding a bias immediately
increases the coverage to 7.  This pattern is consistent for a wide
range of values as shown in the Table.

We also compute the KL-divergence between transactions in testing
period and predictions.  As some items appear only in transactions or
predictions, we smooth the distributions by considering all zero
frequency as $0.01$ to avoid an unbounded divergence value.  For the
data in regular season, the KL-divergence between predictions and true
transactions is $0.9413$, and decreases to $0.9406$ once biases are
added. Similarly for data in holiday season, the divergences reduces
from $0.7263$ to $0.6757$ owning to learned biases.

In summary, our proposed method is able to learn biases that
incorporate temporal dynamics of items, both at individual-level and
macro-level, and hence improving recommendation performance.  This
improvement is more observable when there is a big difference between
data distribution for training and recommendation time.

\subsection{Optimizing Different Performance Metrics}

Here we apply the 
proposed bias learning algorithm with respect to rank-related metrics
like MAP or NDCG. 
Figures~\ref{fig:regular-map} and \ref{fig:holiday-map} plot the lift of $M_{bias}$
over baseline $M_{long}$ 
at both data
sets.  First of all, all methods lead to a positive lift, implying the
effectiveness of learning biases.  Nevertheless, the optimization
criteria results in final performance difference.  Initially, we conjecture that
optimizing one metric would lead to higher numbers in terms of that
particular metric, e.g., optimizing MAP should result in higher MAP in
test performance.  In reality, this is not the case.  As seen in
Figure~\ref{fig:holiday-map}, optimizing MAP actually results in lower
performance in terms of all three metrics, and optimizing NDCG, on the
contrary, yields comparable performance to ACC. 

We also compare the difference of learned biases of optimizing ACC or NDCG.
It turns out we learned 386 non-zero biases for ACC, 402 non-zero
biases for NDCG and 384 are shared between both.  The corresponding
bias values are shown in Figure~\ref{fig:bias-values}, with items
sorted based on bias values of optimizing ACC.  It is noticed that
majority of them have close-to-zero values. Though some values are
positive, more bias values are towards negative, suggesting many item
biases tries to discount the popularity in the training data.  The
scale of negative values tend to be much larger.  Moreover, optimizing
NDCG  gives less extreme values, because it considers all the
possible ranking positions of top-$k$.

Recall that the time complexity of each iteration in optimizing NDCG
is $k$ times larger than that of optimizing ACC.  Therefore, we have
to strike a balance between the performance and computational time. 
In practice, we can have one trained
recommendation model, and just need to update the biases daily. Faster
convergence can be accomplished via warm start, i.e., adopting the
learned biases yesterday for initialization.  For our experiment
purpose, we just report findings of optimizing ACC thereafter.


\subsection{Bias Learning for Different Base Models}
In previous subsections, we have mainly studied properties of bias
learning with the Markov model.  Here we explore
other base recommendation models.  Our proposed bias learning is kind
of orthogonal to the base model being used and it is applicable
to a wide array of models.  Two models are considered here: 
matrix factorization (MF) and category-based
recommendation (CBR).

\emph{Matrix factorization} (MF) gained momentum thanks to the Netflix
prize competition\cite{Kore-etal09-matrix,Taka-etal09-scalable}. It is
shown to be one of the start-of-the-art methods for collaborative
filtering. Standard matrix factorization aims to approximate a
user-item matrix as the product of two low-rank matrices:
$A_{n\times m} \approx P_{n\times \ell } Q_{\ell \times m}.$
where $P$ and $Q$ are the latent factors
of users and items, respectively.
Typical matrix factorization (either through alternating least squares
or stochastic gradient descent) requires an iterative process, which
involves too much overhead when implemented in MapReduce in order to
deal with large-scale data sets.  Alternatively, we implemented a
randomized version of matrix factorization as described in
\cite{Tang-Harr13-scaling}.   It utilizes a randomized
SVD~\cite{Halk-etal11-finding} to compute approximate $Q$ and then
determines $P$ given $Q$.

\emph{Category-based recommendation} (CBR) assumes a user is more likely to
purchase a product within the same category if he has already
indicated an interest in a category. 
\begin{displaymath}
 P(\text{buy i} | \text{user}~u )  = \sum \limits_{c} P(\text{buy i} | \text{buy in}~c) 
  \cdot P(\text{buy in}~c|\text{user}~u). 
\end{displaymath}
The interest categories of one user $P(\text{buy in}~c|\text{user}~u)$
is computed through his past actions.  Each user is represented as a
multinomial distribution of interest categories, by mapping each of
his actions to its corresponding category in a carefully curated
product taxonomy.  To estimate $P(\text{buy i} | \text{buy in}~c)$, we
examine the popularity of each product among existing transactions. In
order to capture the recent trend, we apply $M_{truncate}$, that is,
we restrict ourselves to look at transactions only within the past
few days/weeks at recommendation time.  Therefore, this category-based
recommendation tends to pick those recent best-selling products given
one's personal interest categories.


We apply bias learning to both models.  For brevity, we just report
the lift in terms of ACC@10 over different base models in
Table~\ref{tab:different-base}.  For all base models, we observe a
positive lift, suggesting that our bias learning is able to capture
the temporal dynamics no matter which model is being used.  The lift
over MF is substantial, partly because of MF's poor performance
itself.  To our surprise, among all three methods,
matrix factorization performs the worst.  Such a poor performance of
matrix factorization is also observed in other domains with binary
responses~\cite{Aiol13-efficient}.  One factor is that matrix
construction based on transactions is critical yet not well
defined. It is difficult to incorporate both frequency and recency
information simultaneously into one matrix.

\begin{table}[t]
  \centering
  \caption{Lift(\%) of Bias Learning over Base Models}
  \label{tab:different-base}
  \begin{tabular}{c|ccc}
    \hline
    base model & MF & CBR & Markov\\
    \hline
    regular-season & 39.100 & 0.425 &  1.228\\
    holiday-season & 27.938 & 5.477 & 5.857\\
    \hline
  \end{tabular}
  \vskip -0.2in
\end{table}

\subsection{Online A/B Tests}
Here we run online A/B test through email campaigns to further examine
the impact of added bias for recommendation.  Each email contains 8
item recommendations.  As mentioned earlier, the Markov model works
quite well in our domain and is adopted for base recommendation.
We sample a small percentage of XYZ customers and randomly split
them into two buckets, one with bias learning and the other without.
We run two tests, on 2013/11/26 and 2013/12/07, respectively. 
Both tests sent around 800K
marketing emails. 
Three widely-used
metrics are recorded: the click-through rate
(CTR), 
average number of orders and revenue per email-open. We attribute one
order/revenue to be from marketing emails only if customers receiving
emails click on one link in the email and place order(s) within the same
session.
\begin{table}[b]
  \centering
\vskip -0.3in
  \caption{Lift(\%) of online recommendation with bias}
  \label{tab:online}
  \begin{tabular}{c|ccc}
    \hline
    Date   &  CTR & avg \#order & avg revenue \\
    \hline 
    2013/11/26  & 3.35 & 6.10 & 5.64 \\
    2013/12/07  & 9.30 & 108.05 & 102.76 \\
    \hline
  \end{tabular}
\end{table}
The lifts after bias learning are shown in
Table~\ref{tab:online}.  For all metrics, we see a positive lift,
though only the CTR is shown to be significant. Because the number of orders
was extremely small, it is difficult to reject the null hypothesis given limited impressions.
These online
tests confirm our hypothesis that adding a bias to capture temporal
dynamics intrigue more customers to click and place orders
subsequently, suggesting more effective recommendation.

\section{Related Work}
Mining concept-drifting data streams~\cite{Wang:2003:MCD:956750.956778} for classification and pattern
mining has been studied extensively. 
Yet considering temporal changes for recommendation is gaining some
attention recently.  Koren et al.~\cite{Kore09-collaborative,Koen-etal11-Yahoo}
proposed to have a time-dependent bias in matrix factorization for
movie/music ratings. But the proposed method is not applicable for one-class collaborative filtering~\cite{Pan-etal08-one} problem. Moreover, it aims to minimize the root mean squared error (which is differentiable) rather than ranking metrics for top-$k$ recommendation.
On
the other hand, Xiong et al.~\cite{xiong2010temporal} formulate
temporal collaborative filtering as a tensor factorization by
treating time as one additional dimension.  Wang et
al.~\cite{wang2011utilizing,Wang-Zhan13-opportunity} 
consider the time gap between purchases and propose an opportunity
model to identify not only the items to recommend, but also the best
timing to recommend a particular product. Meanwhile, improving
temporal diversity of recommendation across
time~\cite{Lath-etal10-temporal,Zhao-etal12-increasing} is also
considered.

Another related domain is learning to rank~\cite{liu2009learning},
which is initially motivated for the problem of information retrieval
given queries.  Making recommendations by learning to rank has
attracted lots of attentions
recently~\cite{rendle2009bpr,weston2013learning,Kara-etal13-learning}.
EigenRank~\cite{Liu-Yang08-eigenrank} extends memory-based (or
similarity-based) methods by considering the ranking (rather than
rating) of items in computing user similarities.  Matrix factorization
has been extended to optimize for ranking-oriented loss as well.  But
most ranking-related metrics are non-smooth or non-convex. Hence,
majority of the methods either approximate the loss via a smooth
function or find a smooth lower/upper bound for the loss function.
For instance, CofiRank~\cite{weimer2007cofi} extends
matrix factorization to optimize ranking measures like NDCG instead of
rating measures. Because NDCG is non-convex, the authors propose a
couple of steps to find a a convex upper-bound for the non-convex
problem and adopt bundle method for optimization.
CLiMF~\cite{Shi-etal12-climf} instead optimizes a lower bound of
smooth reciprocal rank.  Our proposed method differs because we
explicitly optimizes for the exact ranking measure. This is viable
because we are learning only biases, rather than latent factors, with
the ranking loss.

Our proposed bias learning method in collaborative
filtering is partly inspired from the thresholding problem in
multi-class/label classification~\cite{Yang01-study,Fan-Jin07-study}.
For large-scale multi-class/label classification problem, one-vs-rest
is still widely used. That is, for each class we construct a binary
classifier by treating the class as positive, and the remaining
classes as negative. Since each binary classifier is constructed
independently, researchers propose to learns a threshold (bias) for
each class mainly to optimize classification accuracy, precision/recall
or F-measure.  
However, in top-$k$
recommendation, the score difference and ranking of items matter,
making all the items dependent on each other.  Also, the motivation
of this work is mainly to capture temporal dynamics rather than
calibrating the classifier prediction
scores. 

\section{Conclusions and Future Work}
This work attempts to take into account temporal dynamics for
top-$k$ recommendation. It is motivated from the observation that
certain domains, e.g.  eCommerce, are highly dynamic.  Since user
feedbacks are likely to be rare in most recommender systems, we
suggest keeping as much data as possible for training recommendation
model to avoid sparsity problem.  On the other hand, we propose to
learn a time-dependent bias for each item based on recent user
feedback only to capture the temporal trend change.  We define the
bias learning problem and present a coordinate-descent like algorithm
to optimize ranking-based measures like ACC, MAP or
NDCG. We prove that the algorithm is guaranteed to terminate in
finite steps with reasonable time complexity.  Empirical results via
both offline and online experiments demonstrate that the proposed bias
learning method is able to boost the performance of base
recommendation models, and capture the temporal shift in user
feedback.  As the bias learning works independently of base
recommendation model being used,  we encourage other practitioners to add
it as a standard module in recommender systems where temporal dynamics
are a norm. 

A couple of problems remain open.  Even though we have
provided some guidelines to reduce computational cost for bias
learning, the proposed algorithm is sequential in nature and thus
difficult to harness the power of parallel/distributed computing. Its
scalability needs to be improved.  For another, it has been shown that  one item can
be removed from recommendation if it does not appear in the recent
transactions. So bias learning would not pick new
items. 
This seems to be against the initial purpose of
recommendation,  to encourage users to discover more items in the long
tail.  It is pressing to understand more about the balance between
relevance, 
popularity and serendipity in recommendation. 

\bibliographystyle{abbrv}
\bibliography{main}

\end{document}